\documentclass[review, 3p]{elsarticle}

\journal{ }
\usepackage{graphicx}
\usepackage{subfigure}
\biboptions{authoryear}
\RequirePackage[colorlinks,citecolor=blue,urlcolor=blue]{hyperref}
\usepackage{cleveref}

\usepackage{amsmath}

\usepackage{times}
\usepackage{bm}
\usepackage{natbib}
\usepackage{amssymb}
\usepackage{setspace}

\def\T{{ \mathrm{\scriptscriptstyle T} }}

\newdefinition{rmk}{Remark}
\newproof{proof}{Proof}
\newproof{pot}{Proof of Theorem \ref{thm2}}
\newtheorem{prop}{Proposition}[section]

\usepackage{tikz}
\usetikzlibrary{fit,positioning,calc}
\tikzset{
  font={\fontsize{11pt}{12}\selectfont}}
\usepackage[ruled,norelsize,vlined]{algorithm2e}

\def\T{{ \mathrm{\scriptscriptstyle T} }}

\begin{document}

\begin{frontmatter}

\title{Bayesian inference on group differences in multivariate categorical data}

\author[]{Massimiliano Russo\corref{mycorrespondingauthor}}
\cortext[mycorrespondingauthor]{Corresponding author}
\ead{russo@stat.unipd.it}

\author[]{Daniele Durante}
\ead{durante@stat.unipd.it}
\author[]{Bruno Scarpa}
\ead{scarpa@stat.unipd.it}
\address{Department of Statistical Sciences, University of Padova, 35121 Padova, Italy}

\begin{abstract}
Multivariate categorical data are common in many fields. We are motivated by election polls studies assessing evidence of changes in voters opinions with their candidates preferences in the 2016 United States Presidential primaries or caucuses. Similar goals arise routinely in several applications, but current  literature lacks a general methodology which combines flexibility, efficiency, and tractability in testing for group differences in multivariate categorical data at different---potentially complex---scales. We address this goal by leveraging a Bayesian representation which factorizes the joint probability mass function for the group variable and the multivariate categorical data as the product of the marginal probabilities for the groups, and the conditional probability mass function of the multivariate categorical data, given the group membership. To enhance flexibility, we define  the conditional probability mass function of the multivariate categorical data via a group-dependent mixture of tensor factorizations, thus facilitating dimensionality reduction and borrowing of information, while providing tractable procedures for computation, and accurate tests assessing global and local group differences. We compare our methods with popular competitors, and discuss improved performance in simulations and in American election polls studies.
\end{abstract}


\end{frontmatter}

\section{Introduction}
\label{sec1}
Multivariate categorical data arise frequently in relevant fields of application. Notable examples include epidemiology  \citep[e.g.][]{landis1988}, psychology \citep[e.g.][]{muthen1981}, social science  \citep[e.g.][]{santos2015}, and business intelligence \citep[e.g.][]{vermunt2004}---among others. In such settings it is increasingly common to observe a vector of categorical responses for each subject, along with a qualitative variable indicating membership to a specific group. For example, in psychological studies a vector of categorical traits is typically measured for each individual, and the focus is on studying  differences in these traits across groups, such as gender or level of education  \citep[e.g.][]{shao2014}. We are specifically motivated by election polls studies measuring changes in voters opinions with their preferences for the Presidential candidates, expressed  in the primaries or caucuses of the 2016 United States Presidential elections. These elections have attracted a considerable interest by the political scientists---mainly due the striking and partially unpredicted results---thereby motivating ongoing attempts to understand the determinants underlying the final outcomes. Most of the available political analyses provide qualitative explanations for the effect of the media, and the effectiveness of the different campaigns and supported policies---among others. Refer to  \citet{lil2016} for a careful summary of the most valuable studies and comments. 

{\setstretch{1.0}
\begin{table}
\centering
\def~{\hphantom{0}}
\caption{Opinions on several political topics collected from voters during the 2016 American
    national elections, along with their preference for Hillary Clinton or Bernie Sanders in the 2016 Democratic Presidential primaries.}{%
\begin{tabular}{llll}
 \\
& VOTER 15 & VOTER 16 &  $\ldots$ \\
\hline
{\bf Vote primaries} $x_i$  & \texttt{Hillary Clinton} & \texttt{Bernie Sanders} & $\ldots$ \\
\hline
 \multicolumn{4}{c}{{\bf Political opinions} ${\boldsymbol y}_i=(y_{i1}, \ldots, y_{ip})^\T$}  \\
\hline
{\textsc{Clinton}} &  &  &      \\
\hline
{\textsc{Feel angry}} & \texttt{Never} & \texttt{Never} & $\ldots$\\
$\ldots$ & $\ldots$ & $\ldots$ & $\ldots$\\
{\textsc{Feel disgusted}} &  \texttt{Never} & \texttt{Never}& $\ldots$\\
\hline
{\textsc{Leadership}} & \texttt{Extremely well} & \texttt{Very well}  & $\ldots$\\
$\ldots$ & $\ldots$& $\ldots$ & $\ldots$ \\
{\textsc{Speaks ming}} & \texttt{Extremely well} & \texttt{Very well} & $\ldots$\\
\hline
{\textsc{Trump}}&  &  &      \\
\hline
{\textsc{Feel angry}} & \texttt{Always} & \texttt{About half the time} & $\ldots$\\
$\ldots$ & $\ldots$ & $\ldots$ & $\ldots$\\
{\textsc{Feel disgusted}} & \texttt{Always} & \texttt{Always} & $\ldots$\\
\hline
{\textsc{Leadership}} &  \texttt{Not well at all} & \texttt{Not well at all}  & $\ldots$\\
$\ldots$ & $\ldots$& $\ldots$ & $\ldots$ \\
{\textsc{Speaks ming}} & \texttt{Extremely well} & \texttt{Very well} & $\ldots$\\
\hline
\end{tabular}}
\label{tab_1}
\end{table}}

Although all the above explanations allow important insights, quantitative assessments providing empirical evidence of the suggested conclusions in the light of the observed polls data, are fundamental to improve the current understanding of the determinants underlying the 2016 United States Presidential elections. However, such contributions are still lacking. This is mainly due to the only recent availability of relevant datasets, along with the broad variability of the research interests characterizing the 2016 United States Presidential elections. In this contribution, our overarching goal is to assess evidence of differences in political opinions between the subset of voters who chose Hillary Clinton as Presidential candidate, and the one opting for Bernie Sanders in the 2016 Democratic Presidential primaries. There is, in fact, a common perception in the media that  Bernie Sanders may have been a more effective candidate for the Democratic party in the Presidential campaign against Donald Trump   \citep[e.g.][]{lil2016}. 

As shown in Table \ref{tab_1}, we address the above goal with a main interest on how the voters feelings toward Hillary Clinton and Donald Trump, along with their evaluations on specific personality traits of the two Presidential candidates, change between Hillary Clinton and Bernie Sanders voters  in the 2016 Democratic Presidential primaries. The data are obtained from the American National Election Studies available at {\url{http://electionstudies.org/}}, and comprise five different feelings  along with five specific personality traits for each of the two Presidential candidates, thereby providing a total of $p=20$ categorical opinions collected for each unit. There are $n_1=567$ voters who expressed their preference for Hillary Clinton, and $n_2=386$ voters who chose Bernie Sanders in the 2016  primaries. 

According to Table \ref{tab_1}, it is not clear---a priori---whether there exist group differences in the voters opinions, and, if present, whether these differences  are found in the entire vector of the $p=20$ categorical variables, or only on a subset of the marginals or higher-order structures---including the bivariates, and more complex joint combinations. Obtaining statistical evidence of these differences at multiple scales, can provide interesting insights on how marginal, bivariate, or more complex joint opinions of the voters change with their preference for Hillary Clinton or Bernie Sanders in the 2016 Democratic Presidential primaries. However, as discussed in Section \ref{lit} below, the available literature lacks---to our knowledge---a general methodology to test for group differences in multivariate categorical data at multiple scales under a single statistical model which combines flexibility, efficiency, and tractability. 

To cover this gap, we propose in Section \ref{sec2} a flexible dependent mixture of tensor factorizations, which allows the joint probability mass function for the multivariate categorical data to be unknown, and changing with the groups via a set of group-dependent mixing probabilities. The proposed representation allows substantial dimensionality reduction and efficient borrowing of information in sparse tables, while providing a simple test for global group differences in the entire probability mass function, based on a flexible formulation which reduces model misspecification issues. Taking a Bayesian approach to inference, we define the prior distributions for the parameters in the proposed statistical model to guarantee  full support, and incorporate automatic multiplicity control when testing for local group differences in the marginals, bivariates, and---potentially---more complex combinations. As discussed in Section \ref{sec3}, posterior inference is available via a simple Gibbs sampler which incorporates  the global test, and provides tractable methods for the local tests via a model-based version of the Cramer's \textsc{v} coefficient outlined in Section \ref{test_sec}. The advantages associated with the proposed methods are empirically described via simulations in Section \ref{sec4}, and compared to popular competitors. Finally, Section \ref{sec5} provides results for the motivating application on recent American election polls data.

\subsection{Literature review}
\label{lit}
There is a wide interest in studying differences in political opinions across groups of voters defined by  gender \citep[e.g.][]{atk2003}, race \citep[e.g.][]{brown2009},  and affiliation  party \citep[e.g.][]{finkel1985}---among others. In accomplishing this goal, a widely used approach proceeds by summarizing the multivariate categorical data into a single latent class membership variable, while testing for group differences in these latent classes \citep{bolck2004}. Although the latent class analysis provides a useful simplification, the set of procedures required to perform the above test are subject to systematic bias, and it is still an active area of research to improve this methodology \citep[e.g.][]{vermunt2010}.

An alternative procedure is to avoid data reduction by assessing evidence of group differences in each categorical variable via separate $\chi^2$ tests, while accounting for multiple testing via false discovery rate control \citep[e.g][]{benjamini1995}. These methodologies do not incorporate dependence structures among the $p$ categorical variables, and therefore have low power. \citet{pesarin2010} addressed this issue via permutation tests preserving the dependence structure in the multivariate categorical data. Although this contribution provides a possible solution, the proposed methods cannot capture differences that go beyond changes in the marginals of a multivariate categorical variable, thus leading to inaccurate insights when the group differences are in higher-order structures.

To avoid the above robustness issues, one possibility is to define a test based on a provably flexible representation for the probability mass function of the multivariate categorical data. Log-linear models  \citep[e.g.][]{agresti2013} represent a popular class of procedures, but are characterized by an explosion in the set of possible interactions when the number of variables increases. Indeed, even in moderate $p$ settings, the number of parameters required to fully characterize the joint probability mass function is massively larger than the sample size $n$, thereby leading to inaccurate inference on group differences in the entire joint probability mass function, and higher--order dependence structures among the categorical variables. To explicitly incorporate sparsity in log-linear models, \citet{ntzoufras2000}, and \citet{nardi2008} proposed a Bayesian stochastic search, and a group-lasso algorithm, respectively, for model selection. However, when  $p$ is moderate to large, these procedures require restrictions for computational tractability, potentially affecting flexibility. As a result inference can be a cumbersome task in high-dimensional settings.

In order to address the issues associated with log-linear models, an alternative recent literature avoids pre-specifying graphical parametric structures in the multivariate categorical data, but leverages instead tensor factorizations.  \citet{dunsonxing2009} proposed a Bayesian nonparametric representation which defines the probability mass function for the multivariate categorical random vector as a mixture of products of multinomial distributions. This factorization induces a provably flexible statistical model, which incorporates dimensionality reduction, and borrows information between the cell probabilities in sparse tables to provide efficient inference on the entire joint probability mass function. Notable recent generalizations of the model proposed by \citet{dunsonxing2009} incorporate additional sparsity \citep{Zhou2014}, dynamic patterns  \citep{kunihama2013}, classification of univariate outcomes \citep{yang2015}, data imputation \citep{fosdick2016, mur2016}, and inference in case-control studies  with several categorical predictors  \citep{zhou2015}. Refer to \cite{john2014} for a theoretical justification, and connections with log-linear models.

\citet{yang2015}, and \citet{zhou2015} focused on the conditional distribution of a univariate response, with the categorical data acting as predictors. Consistent with the discussion in Section \ref{sec1}, we consider instead the dual problem, assessing evidence of group differences in the entire probability mass function of a multivariate categorical random variable. This is accomplished by factorizing the joint probability mass function for the group variable and the multivariate categorical data, as the product of the marginal probabilities for the groups and the conditional probability mass function of the multivariate categorical data, defined via a group-dependent mixture of tensor factorizations. As discussed in Section \ref{sec2}, this formulation is flexible and tractable, facilitating accurate global and local testing.

\section{Dependent mixture of tensor factorizations for multivariate categorical data}\label{sec2}
Let  ${\boldsymbol{y}}_i=(y_{i1}, \ldots, y_{ip})^{\T} \in \mathcal{Y} =(1, \ldots, d_1) \times \cdots \times (1, \ldots, d_p)$ denote the vector of categorical data observed for the statistical unit $i$, and $x_i \in \mathcal{X} = (1, \ldots, d_x )$ its corresponding group, for  each unit $i=1, \ldots, n$. We seek a provably flexible representation for the joint probability mass function $\boldsymbol{\pi}_{\boldsymbol{Y},X}=\{\pi_{\boldsymbol{Y},X}(\boldsymbol{y},x)=\mbox{pr}(\boldsymbol{Y}=\boldsymbol{y},X=x): \boldsymbol{y} \in   \mathcal{Y}, x  \in  \mathcal{X}\}$ underlying data $(\boldsymbol{y}_1,x_1), \ldots, (\boldsymbol{y}_n,x_n)$,  which facilitates accurate testing of independence between the random variables $\boldsymbol{Y}$ and $X$. This goal can be formally addressed by assessing evidence against the null hypothesis
\begin{eqnarray}
H_0: \pi_{\boldsymbol{Y},X}(\boldsymbol{y},x) =  \pi_{\boldsymbol{Y}}(\boldsymbol{y})\pi_{X}(x), \quad \mbox{for all $\boldsymbol{y} \in   \mathcal{Y}$ and  $x  \in  \mathcal{X}$},
\label{eq1}
\end{eqnarray}
versus the alternative
\begin{eqnarray}
\ \  \ H_1: \pi_{\boldsymbol{Y},X}(\boldsymbol{y},x) \neq  \pi_{\boldsymbol{Y}}(\boldsymbol{y})\pi_{X}(x), \quad \mbox{for some $\boldsymbol{y} \in   \mathcal{Y}$ and  $x  \in  \mathcal{X}$},
\label{eq2}
\end{eqnarray}
with  $\boldsymbol{\pi}_{\boldsymbol{Y}}=\{\pi_{\boldsymbol{Y}}(\boldsymbol{y})=\mbox{pr}(\boldsymbol{Y}=\boldsymbol{y}): \boldsymbol{y} \in   \mathcal{Y}\}$ and $\boldsymbol{\pi}_{X}=\{\pi_{X}(x)=\mbox{pr}(X=x): x  \in  \mathcal{X}\}$ denoting the unconditional probability mass functions of $\boldsymbol{Y}$ and $X$, respectively.

In order to accurately test the system of hypotheses \eqref{eq1}--\eqref{eq2}, and avoid issues arising from model misspecification, it is important to develop a representation for $\boldsymbol{\pi}_{\boldsymbol{Y},X}$ which is sufficiently general to approximate any possible probability mass function in the $|\mathcal{Y} \times \mathcal{X}|-1$ dimensional simplex $\mathcal{P}_{|\mathcal{Y} \times \mathcal{X}|-1}$. For instance, restrictive representations for the joint probabilistic process associated with $\boldsymbol{Y}$, are expected to fail in detecting group differences in complex structures of the multivariate categorical random variable, beyond those incorporated by the assumed statistical model for $\boldsymbol{\pi}_{\boldsymbol{Y},X}$. To avoid this issue, without relying on excessively parameterized models, we express $\pi_{\boldsymbol{Y},X}(\boldsymbol{y},x)$ as
\begin{eqnarray}
\pi_{\boldsymbol{Y},X}(\boldsymbol{y},x)=\pi_{\boldsymbol{Y} \mid X=x}(\boldsymbol{y})\pi_{X}(x), \quad \mbox{for every $\boldsymbol{y} \in   \mathcal{Y}$ and $x  \in  \mathcal{X}$},
\label{eq3}
\end{eqnarray}
with the conditional probability mass function of $\boldsymbol{Y}$ given $X=x$ factorized as a dependent mixture of products of multinomial distributions, obtaining
\begin{eqnarray}
\pi_{\boldsymbol{Y} \mid X=x}(\boldsymbol{y})=  \sum_{h=1}^H \nu_{hx} \prod_{j=1}^p \pi_{hj}(y_j), \quad \mbox{for every $\boldsymbol{y} \in   \mathcal{Y}$ and $x  \in  \mathcal{X}$},
\label{eq4}
\end{eqnarray}
where $\pi_{hj}(y_j)$ is the probability that the categorical random variable $Y_j$ assumes value $y_j$ in mixture component $h$, for each $y_j \in (1, \ldots, d_j)$, $j=1, \ldots, p$ and $h=1, \ldots, H$, while $\boldsymbol{\nu}_x=(\nu_{1x}, \ldots, \nu_{Hx}) \in \mathcal{P}_{H-1}$ are vectors of mixing probabilities specific to each group $x\in (1, \ldots, d_x)$.  Representation \eqref{eq3}--\eqref{eq4} has several benefits. In particular factorization \eqref{eq3} allows inference on changes in the multivariate random variable $\boldsymbol{Y}$ across the groups defined by $X$, with the conditional probability mass functions $\boldsymbol{\pi}_{\boldsymbol{Y} \mid X=x}=\{\pi_{\boldsymbol{Y} \mid X=x}(\boldsymbol{y})=\mbox{pr}(\boldsymbol{Y}=\boldsymbol{y} \mid X=x): \boldsymbol{y} \in   \mathcal{Y}\}$ fully characterizing such variations for each group  $x\in (1, \ldots, d_x)$. Equation \eqref{eq4} generalizes instead unconditional tensor factorization representations  \citep{dunsonxing2009}  to provide a tractable model for the probability mass function of $\boldsymbol{Y}$, which is additionally allowed to flexibly change across the groups $x\in (1, \ldots, d_x)$ via a set of group-specific mixing probabilities. Moreover, the conditional independence assumption among the $p$ categorical variables within each mixture component, allows substantial dimensionality reduction for tractable inference, while incorporating effective borrowing of information. This facilitates modeling of higher-order structures  in sparse tables,  and shrinking towards low-rank representations which are allowed to vary across groups via group-specific mixing probabilities $\boldsymbol{\nu}_x$. 

As discussed in Proposition \ref{prop1}, considering  a conditional independence assumption within each mixture component, and accounting for group-dependence only in the mixing probabilities, do not affect flexibility and incorporates borrowing of information across  the shared mixture components, along with tractable tests of independence between $\boldsymbol{Y}$ and $X$---as we will discuss in Section \ref{test_sec}.

\begin{prop}
 Any  $\boldsymbol{\pi}_{\boldsymbol{Y},X} \in \mathcal{P}_{|\mathcal{Y} \times \mathcal{X}|-1}$ can be factorized as in \eqref{eq3}--\eqref{eq4} for some $H$.
  \label{prop1}
\end{prop}
\begin{proof}
Since it is always possible to factorize $\pi_{\boldsymbol{Y},X}(\boldsymbol{y},x)$ as $\pi_{\boldsymbol{Y},X}(\boldsymbol{y},x)=\pi_{\boldsymbol{Y} \mid X=x}(\boldsymbol{y})\pi_{X}(x)$, for each $\boldsymbol{y} \in   \mathcal{Y}$ and $x  \in  \mathcal{X}$, Proposition \ref{prop1} holds if any collection of conditional probability mass functions $\boldsymbol{\pi}_{\boldsymbol{Y} \mid X=x}=\{\pi_{\boldsymbol{Y} \mid X=x}(\boldsymbol{y})=\mbox{pr}(\boldsymbol{Y}=\boldsymbol{y} \mid X=x): \boldsymbol{y} \in   \mathcal{Y}\}$, admits representation \eqref{eq4}, for every $x \in (1, \ldots, d_x)$.  Adapting Corollary 1 in \citet{dunsonxing2009}, it is always possible to separately represent each $\pi_{\boldsymbol{Y} \mid X=x}(\boldsymbol{y})$ as 
$$\pi_{\boldsymbol{Y} \mid X=x}(\boldsymbol{y})=  \sum_{h_x=1}^{H_x} {\nu}_{h_x} \prod_{j=1}^p \pi_{h_xj}(y_j), \quad \boldsymbol{y} \in \mathcal{Y},$$
for every group $x \in(1, \ldots, d_x)$, with $\pi_{h_xj}(y_j)$ denoting the probability that the categorical random variable $Y_j$ assumes value $y_j$ in mixture component $h_x$, given that $X=x$. Hence, the proof follows after defining each $\boldsymbol{\pi}_{hj}$ for $h=1, \ldots, H$ and $j=1, \ldots, p$ as the sequence of unique component-specific probability mass functions appearing in the above separate factorizations for at least one group $x\in \mathcal{X}$. Consistent with this representation, the associated group-specific mixing probabilities will be $\nu_{hx}= {\nu}_{h_x}$ if $\boldsymbol{\pi}_{hj}=\boldsymbol{\pi}_{h_xj}$, $j=1, \ldots, p$ and $\nu_{hx}=0$, otherwise, proving Proposition \ref{prop1}. \qed
\end{proof}

Proposition \ref{prop1} holds for some $H$, which is typically unknown. However, since the set $\mathcal{Y} \times \mathcal{X}$ has finitely many elements, $H$ admits an upper bound $\bar{H}$ which is finite. Hence, we fix $\bar{H}$ at a conservative threshold, and perform Bayesian inference leveraging priors for $\boldsymbol{\nu}_x$, $x \in (1, \ldots, d_x)$ which allow adaptive deletion of unnecessary mixture components not required to characterize the data \citep[e.g.][]{rousseau2011}. If all the mixture components are occupied after performing posterior computation, this suggests that $\bar{H}$ should be increased. 

Large $\bar{H}$ is typically required in situations when the underlying dependence structure is complex, compared to the sample size. In these contexts \citet{bhattacharya2012}, and  \citet{john2014} proposed a generalization of \citet{dunsonxing2009} which allows the latent class indicator variable to be multivariate, in order to improve flexibility, without necessarily relying on a large number of mixture components. Although our model can be generalized to these representations, we obtained good performance in simulations and applications also under a factorization adapting the model in \citet{dunsonxing2009}. Therefore, we leverage their building-block representation which is interpretable, computationally tractable, and allows simple testing---as discussed in the subsequent Section \ref{test_sec}.

\subsection{Model interpretation and hypothesis testing at different scales}
\label{test_sec}

\begin{figure}[t]
\centering
\begin{tikzpicture}[scale=1.1, transform shape]
\tikzstyle{main}=[circle, minimum size = 13mm, thick, draw =black!80, node distance = 16mm]
\tikzstyle{connect}=[-latex, thick]
\tikzstyle{box}=[rectangle, draw=black!100]
  \node[main, fill = white!100] (theta) {$\nu_{hx}$ };
  \node[main] (z) [right=of theta] {$z_i$};
    \node[main,fill = black!10] (d) [below=of z] {$x_i$ };
        \node[main] (p) [left=of d] {$\boldsymbol{\pi}_X$ };
  \node[main, fill = black!10] (w) [right=of z] {$\boldsymbol{y}_i$ };
    \node[main] (psi) [right=of w] {$\boldsymbol{\pi}_{hj}$};
      \path        (theta) edge [connect] (z)
        (z) edge [connect] (w)
        (d) edge [connect] (z)
                (p) edge [connect] (d)
                      (psi) edge [connect] (w);
  \node[rectangle, inner sep=6.5mm, draw=black!100, fit = (theta)] {};
    \node[rectangle, inner sep=6.5mm, draw=black!100, fit = (psi)] {};
        \node[rectangle, inner sep=12mm, draw=black!100, fit = (psi)] {};
                \node[rectangle, inner sep=12mm, draw=black!100, fit = (theta)] {};
  \node[] at (0.18,-0.97) {{\small{$x \in (1, \ldots, d_x)$}}};
    \node[] at (0.8,-1.7) {{\small{$h=1, \ldots, H$}}};
    \node[] at (9,-1) {{\small{$j=1, \ldots, p$}}};
        \node[] at (9.58,-1.7) {{\small{$h=1, \ldots, H$}}};
\end{tikzpicture}
\caption{Graphical representation of the mechanism to generate data $(\boldsymbol{y}_i,x_i)$ from model \eqref{eq3}--\eqref{eq4}.}\label{F_1}
\end{figure}
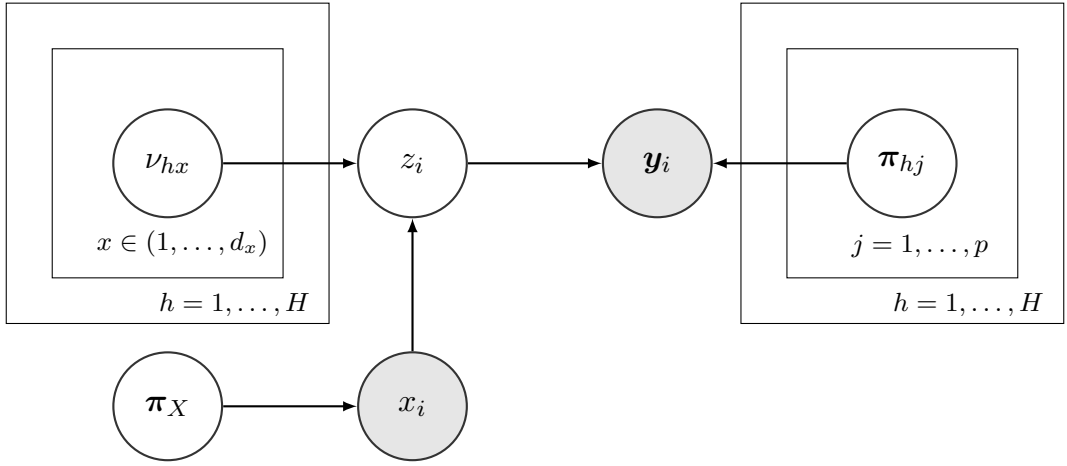

Figure \ref{F_1} provides a graphical representation of the probabilistic generative mechanism associated with our representation of the joint probability mass function $\boldsymbol{\pi}_{\boldsymbol{Y},X}$ via \eqref{eq3}--\eqref{eq4}. According to Figure \ref{F_1}, the group membership $x_i$ is simply generated from the univariate categorical random variable $X$ with unconditional probability mass function $\boldsymbol{\pi}_X$. Conditionally on the group membership $x_i$, the data $\boldsymbol{y}_i$ are instead generated from the multivariate random vector with conditional probability mass function $\boldsymbol{\pi}_{\boldsymbol{Y} \mid X=x_i}$ factorized as in \eqref{eq4}. In accomplishing this goal, a latent class variable $z_i \in (1, \ldots, H)$ is first generated from a categorical variable with probability mass function $\boldsymbol{\nu}_{x_i}$. Then, given $z_i=h$, the entries $y_{ij}$ of $\boldsymbol{y}_i$ are generated from conditionally independent categorical random variables with probability mass function $\boldsymbol{\pi}_{hj}=\{\pi_{hj}(y_j): y_j \in (1, \ldots, d_j)\}$ for every $j=1, \ldots, p$. This hierarchical representation provides key benefits in terms of computational tractability as discussed in Section \ref{post_comp}, while substantially reducing dimensionality, and providing a simple test for global group differences in the multivariate categorical random variable. In fact, it is easy to show that under  \eqref{eq3}--\eqref{eq4} for  $\boldsymbol{\pi}_{\boldsymbol{Y},X}$,  the system \eqref{eq1}--\eqref{eq2} reduces to the simpler test  assessing evidence against 
\begin{eqnarray}
H_0: (\nu_{11}, \ldots, \nu_{H1}) = \cdots  =(\nu_{1d_x}, \ldots, \nu_{Hd_x}),
\label{eq5}
\end{eqnarray}
versus the alternative
\begin{eqnarray}
H_1:  (\nu_{1x}, \ldots, \nu_{Hx})  \neq (\nu_{1x'}, \ldots, \nu_{Hx'}), \ \ \ \ \ \ \
\label{eq6}
\end{eqnarray}
for some $x \neq x'$. This test substantially improves tractability, without affecting accuracy. In fact, according to the aforementioned Proposition \ref{prop1}, the system \eqref{eq5}--\eqref{eq6} leverages a representation of  $\boldsymbol{\pi}_{\boldsymbol{Y},X}$   which is provably general, and therefore reduces concerns arising from model misspecification. 

Rejection of the global null in the system \eqref{eq5}--\eqref{eq6} provides evidence of group differences in the multivariate categorical random variable $\boldsymbol{Y}$. However, such changes may be attributable to several structures.  Consistent with this discussion, we additionally consider local analyses assessing evidence of group differences in each marginal $Y_j$ of $\boldsymbol{Y}$, $j=1, \ldots, p$, and in the bivariates of each  pair $(Y_j,Y_{j'})$, for every $j=1, \ldots, p$ and $j'=1, \ldots,p$, $j'\neq j$. 

We address this aim by relying on a test which leverages the model-based version of the Cramer's \textsc{v} coefficient. Specifically, we assess evidence of group differences in each marginal $Y_j$, for $j=1, \ldots, p$, by studying the coefficients
\begin{eqnarray}
\rho_j=\left[\frac 1  {\min\{d_x,d_j\}-1} \sum_{x=1}^{d_x} \sum_{y_j=1}^{d_j} \frac{\{\pi_{Y_j,X}(y_j,x) - \pi_{Y_j}(y_j) \pi_{X}(x) \}^2 } { \pi_{Y_j}(y_j) \pi_{X}(x) }\right]^{\frac{1}{2}},\label{eq8}
\end{eqnarray}
for each $j=1, \ldots, p$, where $\pi_{Y_j}(y_j)$ denotes $\mbox{pr}(Y_j=y_j)$, whereas $\pi_{Y_j,X}(y_j,x)=\mbox{pr}(Y_j=y_j,X=x)=\mbox{pr}(Y_j=y_j \mid X=x)\mbox{pr}(X=x)=\pi_{Y_j \mid X=x}(y_j) \pi_{X}(x)$ for every $y_j \in (1, \ldots, d_j)$ and group $x \in (1, \ldots, d_x)$. Measuring the association between $Y_j$ and $X$ with $\rho_j \in [0,1]$ provides a convenient choice for interpretation. In fact, according to \eqref{eq8}, a value of $\rho_j$ very close to $0$ provides evidence of low dependence between $Y_j$ and $X$, meaning that $Y_j$ is not expected to change across groups.

We consider a similar strategy to study group differences in the bivariate probability mass functions for every pair $(Y_j,Y_{j'})$  across the categories of $X$.  As in equation \eqref{eq8}, this is accomplished by studying the coefficient
\begin{eqnarray}
\rho_{jj'} &=& \left[
\frac{1}{\min\{d_x,d_j d_{j'}\} -1}
\sum_{x=1}^{d_x} \sum_{y_j =1}^{d_j}\sum_{y_{j'} =1}^{d_{j'}}\frac{
\{ \pi_{Y_j,Y_{j'},X}(y_j,y_{j'},x) -
\pi_{Y_j,Y_{j'}}(y_j,y_{j'})\pi_{X}(x)
\}^2
}
{
\pi_{Y_j,Y_{j'}}(y_j,y_{j'})\pi_{X}(x)
} \right]^{\frac 12},
\label{eq9}
\end{eqnarray}
for every $x \in (1, \ldots, d_x)$ and pair $(Y_j,Y_{j'})$, $j=1, \ldots, p$, $j'=1, \ldots,p$, $j'\neq j$. In equation~\eqref{eq9}, $\pi_{Y_j ,Y_{j'}, X}(y_j,y_{j'},x)$ denotes the joint probability that $Y_j$, $Y_{j'}$ and $X$ take values $y_j$, $y_{j'}$ and $x$, respectively,  while $\pi_{Y_j,Y_{j'} }(y_j,y_{j'})$ is the joint probability of  the pair $(y_j,y_{j'})$, and $\pi_X(x)$ the marginal probability to observe the group $x$. Consistent with the above discussion, a value of $\rho_{jj'}$ very close to $0$ suggests low evidence of changes in $(Y_j,Y_{j'})$ with $X$.

Beside providing simple measures for interpretable inference on local group differences, the above model-based Cramer's \textsc{v} coefficients incorporate dependence in the multiple tests via the factorization  \eqref{eq3}--\eqref{eq4}, and therefore are expected to improve power. Moreover, according to Proposition \ref{prop3}, the Cramer's \textsc{v} coefficients in \eqref{eq8} and \eqref{eq9} can be easily computed from the quantities in our model, facilitating tractable testing at multiple scales under a single model. 
\begin{prop}
Let $\mathcal{J} \subset (1, \ldots, p)$ denotes a generic subset of the indices set  $(1, \ldots, p)$, such that  $\mathcal{J} \cup \mathcal{J}^{c}=(1, \ldots, p)$, and let $\boldsymbol{Y}_{\mathcal{J}}$ denote the multivariate categorical random vector containing the variables with indices in the set $\mathcal{J}$. Then, under the factorization \eqref{eq3}--\eqref{eq4} for $\pi_{\boldsymbol{Y},X}(\boldsymbol{y},x)$, we obtain  $\pi_{\boldsymbol{Y}_{\mathcal{J}} \mid X=x}(\boldsymbol{y}_{\mathcal{J}})=\sum_{h=1}^H \nu_{hx} \prod_{j \in \mathcal{J} } \pi_{hj}(y_j)$, and $\pi_{\boldsymbol{Y}_{\mathcal{J}}}(\boldsymbol{y}_{\mathcal{J}})=\sum_{x \in \mathcal{X}} \pi_X(x)  \{\sum_{h=1}^H \nu_{hx} \prod_{j \in \mathcal{J} } \pi_{hj}(y_j)\}$.
  \label{prop3}
\end{prop}
\begin{proof}
To obtain $\pi_{\boldsymbol{Y}_{\mathcal{J}} \mid X=x}(\boldsymbol{y}_{\mathcal{J}})$ we need to marginalize out in $\pi_{\boldsymbol{Y} \mid X=x}(\boldsymbol{y})$ all the configurations $\boldsymbol{y}_{\mathcal{J}^c} \in \mathcal{Y}_{\mathcal{J}^c}$. To accomplish this goal note that $\boldsymbol{Y}=(\boldsymbol{Y}_{\mathcal{J}},\boldsymbol{Y}_{\mathcal{J}^c})$ and $\boldsymbol{y}=(\boldsymbol{y}_{\mathcal{J}},\boldsymbol{y}_{\mathcal{J}^c})$. Therefore, recalling factorization \eqref{eq4}, we obtain
\begin{eqnarray*}
\pi_{\boldsymbol{Y}_{\mathcal{J}} \mid X=x}(\boldsymbol{y}_{\mathcal{J}})={}\sum_{\boldsymbol{y}_{\mathcal{J}^c} \in \mathcal{Y}_{\mathcal{J}^c}} \sum_{h=1}^{H} {\nu}_{hx}{}\prod_{j \in \mathcal{J}} \pi_{hj}(y_j)\prod_{j \in \mathcal{J}^c} \pi_{hj}(y_j)&=&{}\sum_{h=1}^{H} {\nu}_{hx} \prod_{j \in \mathcal{J}} \pi_{hj}(y_j)\left[\sum_{\boldsymbol{y}_{\mathcal{J}^c} \in \mathcal{Y}_{\mathcal{J}^c}}\prod_{j \in \mathcal{J}^c} \pi_{hj}(y_j)\right],\\
&=&\sum_{h=1}^{H} {\nu}_{hx} \prod_{j \in \mathcal{J}} \pi_{hj}(y_j),
\end{eqnarray*}where the last equality follows after noticing that $\prod_{j \in \mathcal{J}^c} \pi_{hj}(y_j)$ is the joint probability mass function of a multivariate categorical random vector with $|\mathcal{J}^c|$ independent variables and joint sample space $\mathcal{Y}_{\mathcal{J}^c}$. Therefore the summation of its joint probability mass function on the whole sample space provides $\sum_{\boldsymbol{y}_{\mathcal{J}^c} \in \mathcal{Y}_{\mathcal{J}^c}}\prod_{j \in \mathcal{J}^c} \pi_{hj}(y_j)=1$. Exploiting model \eqref{eq3} for $\pi_{\boldsymbol{Y},X}(\boldsymbol{y},x)$, the proof of $\pi_{\boldsymbol{Y}_{\mathcal{J}}}(\boldsymbol{y}_{\mathcal{J}})=\sum_{x \in \mathcal{X}} \pi_X(x)  \{\sum_{h=1}^H \nu_{hx} \prod_{j \in \mathcal{J} } \pi_{hj}(y_j)\}$ is an immediate consequence of the above derivations. \qed
\end{proof}

Although Proposition \ref{prop3} facilitates inference on group differences in many complex higher-order functionals of $\boldsymbol{Y}$, we focus on changes in interpretable local structures of relevant interest in these types of analyses. Note also that, to assess statistical evidence of group differences in these local structures we rely on the systems of interval hypotheses $H_{0j}: \rho_{j} \leq \varepsilon$ versus $H_{1j}: \rho_{j} > \varepsilon$, $j=1, \ldots, p$ for the marginals, and $H_{0jj'}: \rho_{jj'} \leq \varepsilon$ versus $H_{1jj'}: \rho_{jj'} > \varepsilon$, $j=1, \ldots, p$, $j'=1, \ldots,p$, $j'\neq j$, for the bivariates, with $\varepsilon$ an appropriately selected small threshold denoting the minimum effect size required to declare the presence of a group difference. Popular thresholds in social science studies are $\varepsilon=0.1$ and $\varepsilon=0.3$, denoting  small and moderate differences, respectively \citep[e.g.][]{king2010}. Since there is not an overall agreement in this choice, we consider an intermediate threshold $\varepsilon=0.2$, and maintain this default setting in our simulations in Section \ref{sec4}, and in the application in Section \ref{sec5}, to assess sensitivity to this choice.

\section{Bayesian inference} \label{sec3}
Although inference and hypothesis testing for the model discussed in Section \ref{sec2} can potentially proceed under different paradigms, we rely on a Bayesian treatment  of the representation \eqref{eq3}--\eqref{eq4}, and the associated testing procedures. This choice is appealing in allowing coherent uncertainty quantification, effective borrowing of information, simple inference via the posterior distribution, along with the possibility to incorporate appropriate prior distributions which facilitate automatic multiplicity control for the local tests \citep[e.g.][]{scott2010}, and adaptation of the model dimensions \citep[e.g.][]{rousseau2011}. Section \ref{prior} describes our prior specification and properties, whereas Section \ref{post_comp} provide a pseudo-code with step-by-step implementation of the tractable Gibbs sampler associated with the proposed statistical model.
 
\subsection{Prior specification and properties} \label{prior}
We define independent priors $\boldsymbol{\pi}_X \sim \Pi_{X}$, $\boldsymbol{\nu}_x \sim \Pi_{\boldsymbol{\nu}}$, $x \in (1, \ldots, d_x)$ and $\boldsymbol{\pi}_{hj} \sim \Pi_{\boldsymbol{\pi}_j}$, $j=1, \ldots, p$,  $h=1, \ldots, H$, for the quantities in  \eqref{eq3}--\eqref{eq4}  to induce a prior $\Pi$ for  $\boldsymbol{\pi}_{\boldsymbol{Y},X}$ which has full support, facilitates tractable posterior inference on the association  between $\boldsymbol{Y}$ and $X$, and incorporates shrinkage along with automatic multiplicity control.

In enhancing computational tractability we let $\Pi_{X}$ and $\Pi_{\boldsymbol{\pi}_j}$, for $j=1,\ldots, p$, correspond to conjugate Dirichlet priors, obtaining $\boldsymbol{\pi}_X \sim \mbox{Dir}(\alpha_1, \ldots, \alpha_{d_x})$, and $\boldsymbol{\pi}_{hj} \sim \mbox{Dir}(\gamma_{j1},\ldots, \gamma_{jd_j})$ independently for $j=1, \ldots, p$, and $h=1, \ldots, H$. The prior $\Pi_{\boldsymbol{\nu}}$ is instead defined to automatically incorporate the global test in  \eqref{eq5}--\eqref{eq6}. \citet{lock2015} recently addressed a related goal in order to test for equality in distribution, with a particular focus on Gaussian mixture models. We adapt their procedure to our conditional tensor factorization, obtaining
\begin{eqnarray}
\boldsymbol{\nu}_{x} &=&(1-T)\boldsymbol{\upsilon}+T\boldsymbol{\upsilon}_{x}, \ \ x \in \mathcal{X}, \nonumber  \\
\boldsymbol{\upsilon} &\sim& \mbox{Dir}(1/H, \ldots, 1/H), \quad \boldsymbol{\upsilon}_{x} \sim  \mbox{Dir}(1/H, \ldots, 1/H), \ x \in \mathcal{X},\label{eq7} \\
T &\sim& \mbox{Bern}\{\mbox{pr}(H_1)\}. \nonumber 
\end{eqnarray}
According to equations \eqref{eq7}, when $T=0$ the mixing probability vectors are forced to be equal across all groups, while if $T=1$ these vectors are allowed to be different. As shown in \cite{lock2015}, combining this prior with a flexible characterization for the kernels in the mixture model---as in our formulation---provides a provably accurate test for the equality in distribution under a general specification for the mixture components, thereby representing a valid candidate also for our methods. Moreover, by choosing small hyperparameters in the Dirichlet priors in  \eqref{eq7} we also facilitate automatic deletion of redundant mixture components \citep{rousseau2011}. 

Leveraging \eqref{eq7}, evidence against the global null hypothesis is available from $\mbox{pr}\{H_1 \mid (\boldsymbol{y}_1,x_1), \ldots, (\boldsymbol{y}_n,x_n)\}=\mbox{pr}\{T=1 \mid (\boldsymbol{y}_1,x_1), \ldots, (\boldsymbol{y}_n,x_n)\}$, which can be easily computed via the Gibbs sampler outlined in Section \ref{post_comp}---refer in particular to step 4 in Algorithm  \ref{algo}. The posterior probabilities for the local alternatives are instead available via $\mbox{pr}\{H_{1j} \mid (\boldsymbol{y}_1,x_1), \ldots, (\boldsymbol{y}_n,x_n)\}=\mbox{pr}\{\rho_{j}> \varepsilon \mid (\boldsymbol{y}_1,x_1), \ldots, (\boldsymbol{y}_n,x_n)\}$, $j=1, \ldots, p$ for the tests on the marginals, and $\mbox{pr}\{H_{1jj'} \mid (\boldsymbol{y}_1,x_1), \ldots, (\boldsymbol{y}_n,x_n)\}=\mbox{pr}\{\rho_{jj'}> \varepsilon \mid (\boldsymbol{y}_1,x_1), \ldots, (\boldsymbol{y}_n,x_n)\}$, $j=1, \ldots, p$, $j'=1, \ldots,p$,  with $j'\neq j$, for the bivariates. Note that, considering small interval local hypotheses defined via a model-based version of the Cramer's \textsc{v} coefficients, allows the proposed model to place a positive probability mass on each local null, with this probability having a prior distribution induced by $ \Pi_{X}$, $\Pi_{\boldsymbol{\nu}}$, and $\Pi_{\boldsymbol{\pi}_j}$, $j=1, \ldots, p$,  via \eqref{eq8}--\eqref{eq9}. According to \citet{scott2010}, these conditions guarantee automatic multiplicity control within a Bayesian framework, thereby providing an additional relevant benefit associated with the proposed methods.

The above discussion is further confirmed by Proposition \ref{prop2}, guaranteeing that the induced prior $\Pi$ for  $ \boldsymbol{\pi}_{ \boldsymbol{Y},X}$ via \eqref{eq3}--\eqref{eq4} has full support in the probability simplex $\mathcal{P}_{|\mathcal{Y} \times \mathcal{X}|-1}$. This is a key result to guarantee the accuracy of our inference procedures, which may display poor performance if $\Pi$ assigns zero probability to a subset of the possible true data generating processes. 
\begin{prop}
Let $ \boldsymbol{\pi}_{ \boldsymbol{Y},X}\sim \Pi$, with $\Pi$ denoting the prior for $\boldsymbol{\pi}_{ \boldsymbol{Y},X}$ induced by $\Pi_{X}$, $\Pi_{\boldsymbol{\nu}}$, $\Pi_{\boldsymbol{\pi}_1}, \ldots,  \Pi_{\boldsymbol{\pi}_p}$ via  \eqref{eq3}--\eqref{eq4}, then $\Pi\{ \boldsymbol{\pi}_{ \boldsymbol{Y},X}:\sum_{ \boldsymbol{y} \in \mathcal{Y}} \sum_{x \in \mathcal{X}} | {\pi}_{ \boldsymbol{Y},X}( \boldsymbol{y},x)- {\pi}^0_{ \boldsymbol{Y},X}( \boldsymbol{y},x) |<\epsilon\}>0$ for any $\epsilon  > 0$, and $ \boldsymbol{\pi}^0_{ \boldsymbol{Y},X} \in \mathcal{P}_{|\mathcal{Y} \times \mathcal{X}|{-}1}$.
  \label{prop2}
\end{prop}
\begin{proof}
Recalling Proposition \ref{prop1}, it is always possible to rewrite the $L_1$ distance $\sum_{\boldsymbol{y} \in \mathcal{Y}} \sum_{x \in \mathcal{X}} |\pi_{\boldsymbol{Y},X}(\boldsymbol{y},x)-\pi^0_{\boldsymbol{Y},X}(\boldsymbol{y},x) |$  between $\boldsymbol{\pi}_{\boldsymbol{Y},X}$ and $\boldsymbol{\pi}^0_{\boldsymbol{Y},X}$ as
$$\sum_{\boldsymbol{y} \in \mathcal{Y}} \sum_{x \in \mathcal{X}} | \pi_{X}(x)\sum_{h=1}^{H} {\nu}_{hx} \prod_{j=1}^p \pi_{hj}(y_j)-  \pi^0_{X}(x)\sum_{h=1}^{H} {\nu}^0_{hx} \prod_{j=1}^p \pi^0_{hj}(y_j)|, $$
with ${\nu}^0_{hx} ={\nu}^0_{h_x}$ if $\boldsymbol{\pi}^0_{hj}=\boldsymbol{\pi}^0_{h_xj}$, $j=1, \ldots, p$ and $\nu^0_{hx}=0$, otherwise, for $x \in (1, \ldots, d_x)$. Therefore the prior probability $\Pi\{\boldsymbol{\pi}_{\boldsymbol{Y},X}{:}\sum_{\boldsymbol{y} \in \mathcal{Y}} \sum_{x \in \mathcal{X}} |\pi_{\boldsymbol{Y},X}(\boldsymbol{y},x)-\pi^0_{\boldsymbol{Y},X}(\boldsymbol{y},x) |<\epsilon\}$ assigned to a neighborhood of $\boldsymbol{\pi}^0_{\boldsymbol{Y},X}$ is
\selectfont $${\int}{1}\left\{\sum_{\boldsymbol{y} \in \mathcal{Y}} \sum_{x \in \mathcal{X}} |\pi_{\boldsymbol{Y},X}(\boldsymbol{y},x)-\pi^0_{\boldsymbol{Y},X}(\boldsymbol{y},x) |<\epsilon\right\} {\rm d}\Pi_X(\boldsymbol{\pi}_X){\rm d}\Pi_{\boldsymbol{\nu}}(\boldsymbol{\nu}_x) \prod_{h=1}^H \prod_{j=1}^p{\rm d}\Pi_{\boldsymbol{\pi}_j}(\boldsymbol{\pi}_{hj}),$$with $\pi_{\boldsymbol{Y},X}(\boldsymbol{y},x)$ and $\pi^0_{\boldsymbol{Y},X}(\boldsymbol{y},x)$ factorized as above, and $1\{\cdot\}$ denoting an indicator function. Following  \citet{dunsonxing2009}, a sufficient condition for the above integral to be strictly positive is that all the above priors have full $L_1$ support on their corresponding spaces. As $\Pi_X$ and $\Pi_{\boldsymbol{\pi}_1}, \ldots, \Pi_{\boldsymbol{\pi}_p}$ are Dirichlet priors, by definition $\Pi_X$ has full $L_1$ support on the simplex $\mathcal{P}_{|\mathcal{X}|-1}$, and $\Pi_{\boldsymbol{\pi}_j}$ has full $L_1$ support on the simplex $\mathcal{P}_{d_j-1}$, for each $j=1, \ldots, p$. 

To conclude the proof we need to show that $\mbox{pr}( \sum_{x=1}^{d_x} \sum_{h=1}^H |\nu_{hx}-\nu^0_{hx}|< \epsilon_{\nu} )>0$ for every $ \epsilon_{\nu} >0$, and $(\boldsymbol{\nu}^0_{1},\ldots,\boldsymbol{\nu}^0_{d_x})$, when the group-specific mixing probabilities $\boldsymbol{\nu}_{1},\ldots,\boldsymbol{\nu}_{d_x}$ have prior $\Pi_{\boldsymbol{\nu}}$ defined as in equation \eqref{eq7}. Marginalizing out the testing indicator $T$, a lower bound for the previous probability is
\begin{eqnarray*}
\mbox{pr}(H_0)\mbox{pr}\left(\sum_{x=1}^{d_x} \sum_{h=1}^H | \upsilon_{h}-\nu^0_{hx}|<\epsilon_{\nu}\right)+\mbox{pr}(H_1) \prod_{x=1}^{d_x} \mbox{pr} \left(\sum_{h=1}^H | \upsilon_{hx}-\nu^0_{hx}|<\frac{\epsilon_{\nu}}{d_x}\right).
\label{app_1}
\end{eqnarray*}
If the true model is generated under independence between $X$ and $\boldsymbol{Y}$, the true mixing probability vectors are constant across groups, and therefore the Dirichlet priors for $\boldsymbol{\upsilon}$ and $\boldsymbol{\upsilon}_x$, $x \in (1, \ldots, d_x)$, ensure the positivity of both summands. When instead the true mixing probability vectors change across groups, the term $\mbox{pr}(H_0)\mbox{pr}(\sum_{x=1}^{d_x} \sum_{h=1}^H | \upsilon_{h}-\nu^0_{hx}|<\epsilon_{\nu})$ is no more guaranteed to be strictly positive. However $\mbox{pr}(H_1)  \prod_{x=1}^{d_x} \mbox{pr}(\sum_{h=1}^H | \upsilon_{hx}-\nu^0_{hx}|<{\epsilon_{\nu}}/{d_x})$ remains positive for every $\epsilon_{\nu}>0$, since under the alternative we assume independent Dirichlet priors $\Pi_{\boldsymbol{\upsilon}_1}, \ldots, \Pi_{\boldsymbol{\upsilon}_{d_x}}$ for the group-specific mixing probability vectors, each one having full $L_1$ support on  $\mathcal{P}_{H-1}$. \qed
\end{proof}

As $ \boldsymbol{\pi}_{ \boldsymbol{Y},X}$ is fully characterized by finitely many parameters $\{ {\pi}_{ \boldsymbol{Y},X}( \boldsymbol{y},x):  \boldsymbol{y} \in \mathcal{Y}, x \in  \mathcal{X}\}$, Proposition \ref{prop2}, also guarantees  that $\Pi\{\boldsymbol{\pi}_{\boldsymbol{Y},X}{:} \sum_{\boldsymbol{y} \in \mathcal{Y}} \sum_{x \in \mathcal{X}} |{\pi}_{\boldsymbol{Y},X}(\boldsymbol{y},x)-{\pi}^0_{\boldsymbol{Y},X}(\boldsymbol{y},x) |<\epsilon \mid (\boldsymbol{y}_1,x_1), \ldots, (\boldsymbol{y}_n,x_n)\}\rightarrow1$ for any $\epsilon>0$, almost surely when $\boldsymbol{\pi}^0_{\boldsymbol{Y},X}$ is the true probability mass function, thereby ensuring also posterior consistency.

\subsection{Posterior computation}
\label{post_comp}
Posterior computation proceeds via a simple and efficient Gibbs sampler, exploiting the hierarchical representation of model  \eqref{eq3}--\eqref{eq4}, outlined in Figure \ref{F_1}. Refer to Algorithm  \ref{algo} for a pseudo-code with detailed steps. Source \texttt{R} code, and tutorial implementations are available at {\url{https://github.com/danieledurante/GroupTensor-Test}}.
{\setstretch{1.0}
\begin{algorithm}[h!]
 \caption{Gibbs sampler for posterior computation} 
 \label{algo}
     \Begin{\vspace{3pt}
  {\bf [1]} Update the marginal probability mass function $\boldsymbol{\pi}_X$ for the group variable $X$, from the full conditional $(\boldsymbol{\pi}_X \mid -) \sim\mbox{Dir}(\alpha_1+n_{1}, \ldots, \alpha_{d_x}+n_{d_x})$, with the generic $n_x$ denoting the total number of statistical units in group $x \in (1, \ldots, d_x)$\;
\vspace{3pt}
 {\bf [2]}  Sample the latent class indicator variables $z_i \in (1, \ldots, H)$ for each unit $i$\;
 \For(){$i$ \mbox{from} $1$ to $n$}
 {
    Sample  $z_i \in (1, \ldots, H)$ from the categorical variable with probabilities $$\mbox{pr}(z_i=h \mid -)=\frac{\nu_{hx_i} \prod_{j=1}^p \pi_{hj}(y_{ij})}{\sum_{q=1}^H \nu_{qx_i} \prod_{j=1}^p \pi_{qj}(y_{ij})}, $$
for every $h=1, \ldots, H$.    }
\vspace{3pt}
 {\bf [3]}  Update the component-specific probability mass functions $\boldsymbol{\pi}_{hj}$ in equation \eqref{eq4}\;
 \For(){$h$ \mbox{from} $1$ to $H$}
 { \For(){$j$ \mbox{from} $1$ to $p$}
 {Update $\boldsymbol{\pi}_{hj}$ from  $(\boldsymbol{\pi}_{hj} \mid -)\sim\mbox{Dir}(\gamma_{j1}+n_{jh1}, \ldots, \gamma_{jd_j}+n_{jhd_j})$, with the generic $n_{jhy_j}$ denoting the number of statistical units in component $h$ having value $y_j$ for the variable $Y_j$.} }
\vspace{3pt}
{\bf [4]} Sample the testing indicator $T$ from the full conditional Bernoulli variable with
\begin{eqnarray*}
\mbox{pr}(T=1 \mid -) &=&\frac{\mbox{pr}(H_1)\prod_{x=1}^{d_x}\int (\prod_{h=1}^{H}\upsilon_{hx}^{n_{hx}})d \Pi_{\upsilon_x} }{\mbox{pr}(H_0)\int (\prod_{h=1}^{H}\upsilon_h^{n_h})d \Pi_\upsilon +\mbox{pr}(H_1)\prod_{x=1}^{d_x}\int (\prod_{h=1}^{H}\upsilon_{hx}^{n_{hx}})d \Pi_{\upsilon_x} },\nonumber\\
&=&\left[ 1+\frac{\mbox{pr}(H_0)}{\mbox{pr}(H_1)}\frac{\prod_{h=1}^H\Gamma(\frac{1}{H}+n_h)}{\Gamma(\frac{1}{H})^H\Gamma(n+1)} \prod_{x=1}^{d_x} \frac{\Gamma(\frac{1}{H})^H\Gamma(n_x+1)}{\prod_{h=1}^H\Gamma(\frac{1}{H}+n_{hx})}\right]^{-1},
\end{eqnarray*}
where $n_h$ is the total number of units in mixture component $h$, and $n_{hx}$ is the total number units in group $x$ allocated to component $h$. The above equation can be easily obtained adapting derivations in \cite{lock2015}. Exploiting the Gibbs samples for $T$, the posterior probability of the global alternative can be easily obtained as the proportion of samples in which $T=1$\;
\vspace{3pt}
{\bf [5]} Update the group-specific mixing probability vectors  $\boldsymbol{\nu}_{x}$, $x \in (1, \ldots,d_x)$\;
 \uIf{T=1}{Update each group-specific mixing probability vector $\boldsymbol{\nu}_{x}$ separately from the full conditional $(\boldsymbol{\nu}_{x} \mid -)\sim \mbox{Dir}(1/H+n_{1x}, \ldots, 1/H+n_{Hx})$, for $x \in (1, \ldots, d_x)$}
    \ElseIf{T=0}{Let $\boldsymbol{\nu}_{1}=\cdots=\boldsymbol{\nu}_{d_x}=\boldsymbol{\upsilon}$, with $\boldsymbol{\upsilon}$ updated from the full conditional distribution $(\boldsymbol{\upsilon} \mid -) \sim \mbox{Dir}(1/H+n_1, \ldots, 1/H+n_H)$\;}
    }
\end{algorithm}}

\section{Simulation study}\label{sec4}
We consider three relevant simulation studies to evaluate the empirical performance of the proposed methodologies in several scenarios, characterized by different types of dependence between $\boldsymbol{Y}$ and $X$. In particular, in a first scenario we generate the data to obtain sparse dependence structures in $\boldsymbol{Y}$, with these higher-order dependencies, along with the induced marginals,  being the same across the two groups defined by the variable $X$. The second scenario induces instead dependence between $\boldsymbol{Y}$ and $X$, by incorporating group differences in the marginals of $\boldsymbol{Y}$, along with variations in more complex higher-order structures, including a subset of the bivariates. Finally, the third scenario characterizes a challenging situation in which there are no changes in the marginals of $\boldsymbol{Y}$,  but only sparse group differences in the bivariates. Hence, the dependence between $\boldsymbol{Y}$ and $X$ is in fewer higher-order structures. The goal in defining these challenging simulation scenarios is to assess whether the proposed model can characterize probabilistic generative mechanisms having different properties, thereby ensuring accurate testing in broad settings. Consistent with this goal, we focus on $d_x=2$ groups, and $p=15$ categorical variables having $d_1=\cdots=d_{15}=4$ possible categories. Data  $(\boldsymbol{y}_i,x_i)$, $i=1, \ldots, n$, are simulated for $n=400$ units, whose group membership $x_i$ is generated from a categorical variable with probabilities $\pi_X^0(1)=0.5$ and $\pi_X^0(2)=1-\pi_X^0(1)=0.5$. The multivariate categorical responses  $\boldsymbol{y}_i$ are instead simulated from generative mechanisms incorporating the specific properties of the aforementioned scenarios.

In particular, in the first simulation scenario the generative mechanism associated with $\boldsymbol{Y}$ does not change with groups---i.e. $\boldsymbol{\pi}^0_{\boldsymbol{Y} \mid X=x}=\boldsymbol{\pi}^0_{\boldsymbol{Y}}$. However, to evaluate the flexibility of the proposed model, we define a challenging representation for $\boldsymbol{\pi}^0_{\boldsymbol{Y}}$, in which the subset of variables having indices in  $\mathcal{J} = (1,5,10,12,15)$ are generated from a  joint probability mass function with $\mbox{pr}(Y_1 = Y_5 = Y_{10} = Y_{12} = Y_{15}= y) = 0.1$, for each $y\in (1,\ldots,4)$, and the remaining probability mass of $0.6$ assigned in equal proportion to the other $4^5 - 4$  combinations of categories. The variables with indices in $\mathcal{J}^{c}$ are instead simulated independently from their corresponding marginal probability mass function $\boldsymbol{\pi}^0_{Y_j} \sim \mbox{Dir}(10,10,10,10)$. In the second simulation, we induce instead sparse group differences in marginals and bivariates. To incorporate this behavior,  we still simulate variables with indices in $\mathcal{J}^{c}$ independently, but force the marginals  of $Y_2$ and $Y_8$ to change with groups, by letting $\boldsymbol{\pi}^0_{Y_2\mid X=1}=\boldsymbol{\pi}^0_{Y_8\mid X=1}=(0.45,0.45,0.05,0.05)$, and $\boldsymbol{\pi}^0_{Y_2\mid X=2}=\boldsymbol{\pi}^0_{Y_8\mid X=2}=(0.05,0.05,0.45,0.45)$.  The variables with indices in $\mathcal{J}$ are instead generated as in the first scenario for $X=1$. When $X=2$ these variables are instead simulated independently from the marginal probability mass function $\boldsymbol{\pi}^0_{Y_{j \in \mathcal{J}}\mid X=2}=(0.25,0.25,0.25,0.25)$. As a result we  incorporate group differences in the marginals of $Y_2$ and $Y_8$, along with changes in the bivariates for any pair of variables including $Y_2$ or $Y_8$, and any pair $(Y_j,Y_{j'})$, with $j \in \mathcal{J}$, $j' \in \mathcal{J}$. In fact, note that, the joint generative mechanism for the variables with indices in $\mathcal{J}$ ensures that $\boldsymbol{\pi}^0_{Y_{j \in \mathcal{J}}\mid X=1}=(0.25,0.25,0.25,0.25)$. Therefore only the bivariates of these variables change with groups in the second scenario, whereas the marginals remain constant.  Consistent with this discussion, the third scenario  maintains the same generative process, with the exception  of assuming again $\boldsymbol{\pi}^0_{Y_2\mid X=1}=\boldsymbol{\pi}^0_{Y_2\mid X=2}$ and $\boldsymbol{\pi}^0_{Y_8\mid X=1}=\boldsymbol{\pi}^0_{Y_8\mid X=2}$ as in the first scenario. As a result, no group differences in the marginals are observed, and the dependence between $\boldsymbol{Y}$ and $X$ in the third scenario is only due to sparse group differences in the bivariates $(Y_j,Y_{j'})$, with $j \in \mathcal{J}$, $j' \in \mathcal{J}$.

Before studying the empirical performance, it is worth noticing that the above scenarios rely on generative mechanisms not explicitly related to the statistical model proposed in  \eqref{eq3}--\eqref{eq4}, thereby allowing a more effective validation of the flexibility of our methodologies, since the data are not generated from the model described in Section \ref{sec2}. The three scenarios are indeed more closely related to a log-linear model characterized by sparse and higher-order dependence structures, thus providing a challenging setting. To highlight the benefits associated with the proposed methodologies, we compare performance in global testing with the nonparametric approach of   \cite{pesarin2010}, and the latent class models \citep[e.g.][]{bolck2004}---estimating  the latent classes with the \texttt{R} package \texttt{poLCA}. The competitors in local testing are instead separate $\chi^2$ tests with and without false discovery rate control \citep[][]{benjamini1995}. Accurate and tractable inference under a log-linear model would be possible only by including the structure and restrictions of the above scenarios. However, these properties are not known a priori, and the focus of inference is actually learning these structures. Hence, due to the complex higher-order dependencies in the above simulations, an unstructured log-linear model would require a massive amount of parameters to incorporate these structures, thereby leading to intractable and inefficient inference in practice. Hence, we avoid comparison with log-linear models.


\subsection{Performance in global and local testing}
We perform posterior inference under the proposed model \eqref{eq3}--\eqref{eq4} with priors defined in Section \ref{sec3}, setting $\alpha_1=\alpha_2=1/2$, $\gamma_{j1}=\cdots=\gamma_{jd_j}=1/d_j$ for each $j=1, \ldots, p$, and $\mbox{pr}(H_1)=\mbox{pr}(H_0)=0.5$. We maintained these default hyperparameters in all the three simulations to assess sensitivity to prior settings, observing no evidence that posterior inference is sensitive to these hyperparameter's choices. We consider $5000$ Gibbs samples and set a conservative upper bound $\bar{H}=20$, allowing the sparse Dirichlet prior $\Pi_{\boldsymbol{\nu}}$ to adaptively empty redundant mixture components \citep{rousseau2011}. Trace-plots suggest that convergence is reached after a burn-in  of $1000$. We additionally obtain very good mixing, with most of the effective sample sizes for the quantities of interest around $2400$ out of $4000$. 

Using the Gibbs samples for $T$, in the first simulation scenario, we obtain a posterior probability for the global alternative $\hat{\mbox{pr}}\{H_1 \mid (\boldsymbol{y}_1,x_1), \ldots, (\boldsymbol{y}_n,x_n)\}<0.05$\footnote{The estimated posterior probability of $H_1$ can be easily obtained as the relative frequency of the MCMC samples in which $T=1$.}, providing correct evidence of no group differences in the multivariate categorical data. We observe similarly accurate performance for the other two simulation scenarios, providing a posterior probability  $\hat{\mbox{pr}}\{H_1 \mid (\boldsymbol{y}_1,x_1), \ldots, (\boldsymbol{y}_n,x_n)\}>0.95$, which correctly highlights the global dependence between $\boldsymbol{Y}$ and $X$ in both scenarios. The permutation test proposed in  \cite{pesarin2010} provided correct results in the first two scenarios---when independence and dependence are evident from the marginals---but failed to reject $H_0$ with a $p$-value of $0.4$ in the third scenario. This is not surprising, as this procedure aggregates $p$-values of multiple tests assessing evidence of group differences in the marginals---which do not vary with groups in the third scenario. We additionally attempted the global testing procedure based on the latent class analysis  \citep{bolck2004}, estimating  the latent classes with the \texttt{R} package \texttt{poLCA}. Also this approach produced accurate conclusions in the first two scenarios. However, we found the results  quite unstable in the last case. This may be related to the systematic bias associated with this procedure as well as possible convergence issues in the expectation-maximization algorithm.

As shown in Figure \ref{fig:sim1}, our procedure provides also accurate results in assessing local group differences. Consistent with the three generative mechanisms of the simulated data, the posterior distributions for the coefficients $\rho_j$ provide evidence of group differences in the marginals only for  $Y_2$ and $Y_8$ in the second scenario. We obtain, in fact, $\hat{\mbox{pr}}\{\rho_j>0.2\mid (\boldsymbol{y}_1,x_1), \ldots,(\boldsymbol{y}_n,x_n)\}>0.95$\footnote{The estimated posterior probabilities of $H_{1j}$ and $H_{1jj'}$ can be easily obtained as the relative frequencies of the MCMC samples in which $\rho_{j}>0.2$, and $\rho_{jj'}>0.2$, respectively.} only for $j \in (2,8)$, in the second scenario. Similarly accurate performance is found in the local tests on the bivariates. Consistent with the first  scenario, the posterior distribution for the coefficient $\rho_{jj'}$ correctly highlight no group differences in the bivariates, with $\hat{\mbox{pr}}\{\rho_{jj'}>0.2\mid (\boldsymbol{y}_1,x_1), \ldots, (\boldsymbol{y}_n,x_n)\}<0.05$ for all $j=1,\ldots,p$ and $j^\prime = 1,\ldots,p$, with $j^\prime \neq j$. As expected, the changes in the two marginals observed in the second scenario, induce also group differences in the bivariates for pairs of variables including $Y_2$ or $Y_8$. We correctly learn also changes across groups in the joint probability mass function for pairs $(Y_j,Y_{j'})$ with $j \in \mathcal{J}$ and $j' \in \mathcal{J}$, $j'\neq j$, consistent with the settings of the second scenario. The same finding is obtained in the third simulation, correctly providing $\hat{\mbox{pr}}\{\rho_{jj'}>0.2\mid (\boldsymbol{y}_1,x_1), \ldots, (\boldsymbol{y}_n,x_n)\}>0.95$ only for  pairs $(Y_j,Y_{j'})$ with $j \in \mathcal{J}$ and $j' \in \mathcal{J}$, $j'\neq j$.

Local analyses via separate $\chi^2$ tests produced several false positives and false negatives when multiplicity control is not considered. Including a false discovery rate control at a level $0.10$ via \cite{benjamini1995}, improves the results, but still provides one false discovery for the local tests on the bivariates in the second scenario, and one false discovery for the local tests on the marginals in the third scenario. These empirical findings further support the proposed procedures, which gain power by borrowing information across the local tests, and incorporate an automatic multiplicity control via the hierarchical Bayesian formulation \citep[e.g.][]{scott2010}. In fact, according to the above results, the proposed methods effectively control the false discoveries, without requiring additional procedures. 

\begin{figure}[t!]
    \centering
    \includegraphics[width=17cm]{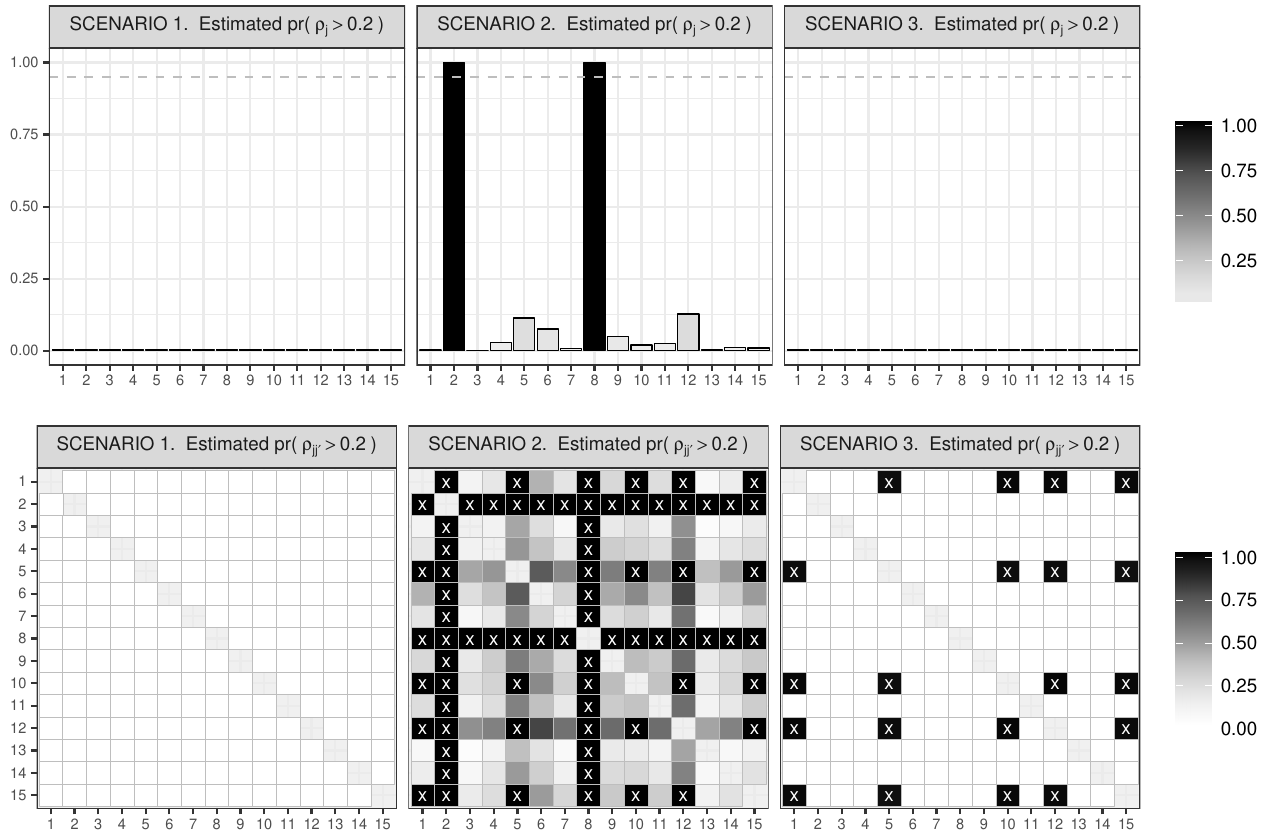}
  \caption{Performance in testing of local group differences. Upper panels: for the three simulation scenarios, posterior estimate of $\mbox{pr}(H_{1j})=\mbox{pr}(\rho_j>0.2)$, to assess evidence of group differences in the marginals $Y_j$, $j=1, \ldots, 15$. Lower panels: for the same scenarios, posterior estimate of $\mbox{pr}(H_{1jj'})=\mbox{pr}(\rho_{jj'}>0.2)$ to test for group differences in the bivariates $(Y_j,Y_{j'})$. The gray dashed lines in the upper panels represent the $0.95$ threshold on the posterior probability of the alternative. The \texttt{x} symbols in the lower panels denote instead those pairs of variables whose bivariates are declared to change across groups according to the proposed local tests---i.e. $\hat{\mbox{pr}}\{\rho_{jj'}>0.2\mid(\boldsymbol{y}_1,x_1), \ldots, (\boldsymbol{y}_n,x_n)\}>0.95$. }
    \vspace{-6pt}
  \label{fig:sim1}
\end{figure}

\section{Application to the 2016 American National Election Studies}\label{sec5}
We apply the proposed methodologies to a subset of the 2016 polls data from the American National Election Studies (\textsc{anes}) available at {\url{http://electionstudies.org/}}, and described in Section~\ref{sec1}. Recalling our motivating application, the dataset comprises $p=20$ categorical measurements of voters opinions and feelings for the two main candidates in the 2016 United States Presidential elections---namely Hillary Clinton and Donald Trump. These categorical data are available on a five item scale, and are collected for $n_1=567$ voters who chose Hillary Clinton during the 2016 Democratic Presidential primaries, and $n_2=386$ voters who expressed preference for Bernie Sanders. Consistent with the discussion in Section~\ref{sec1}, our aim is to understand if the voters feelings and opinions for  Hillary Clinton and Donald Trump, change with their preference for Hillary Clinton or Bernie Sanders expressed in the 2016 Democratic Presidential primaries. Although the  \textsc{anes} dataset provides additional information, and more elaborated analyses could be devised, our fundamental goal is to validate the proposed methods on an interpretable real-data application of potential interest in political studies. Indeed, qualitative political analyses of Presidential primaries are common \citep[e.g.][]{leduc2001,cain2015}, and---as discussed in Section~\ref{sec1}---there is an active debate about the possible effects of a different outcome in the 2016 Democratic Presidential primaries on the final 2016 United States Presidential elections \citep[e.g.][]{lil2016}. 

Focusing on our specific motivating dataset it is not clear---a priori---whether, and for which variables, underlying groups differences are present. In fact, the focus is on democratic voters sharing the same party affiliation. Therefore, their general opinions and feelings toward Hillary Clinton and Donald Trump, may remain substantially unchanged when comparing the subsets of voters expressing their finer--scale preference for one of the two alternative democratic candidates. On the other hand, the substantial differences characterizing the democratic candidates Hillary Clinton and Bernie Sanders \citep[e.g.][]{lil2016}, may have attracted subset of voters with different opinions and feelings toward Hillary Clinton and Donald Trump. However, it is not clear a priori whether the preference for Hillary Clinton and Bernie Sanders, is associated with different positive opinions and feelings for Hillary Clinton or varying negative evaluations of Donald Trump---or both. These considerations motivate the implementation of the statistical model and testing procedures described in Section \ref{sec2}, which are specifically developed to allow effective inference on group differences at varying scales. In accomplishing this goal, we perform posterior computations with the same settings of the simulation studies in Section \ref{sec4}. Also in this case we obtain convergence after a burn-in of $1000$ and good mixing, with most of the effective sample sizes around $2300$ out of $4000$. 

\begin{figure*}[t!] 
\centering
\subfigure[Different feeling towards Hillary Clinton and Donald Trump in the two groups of voters.]{\includegraphics[width=1.01\textwidth]{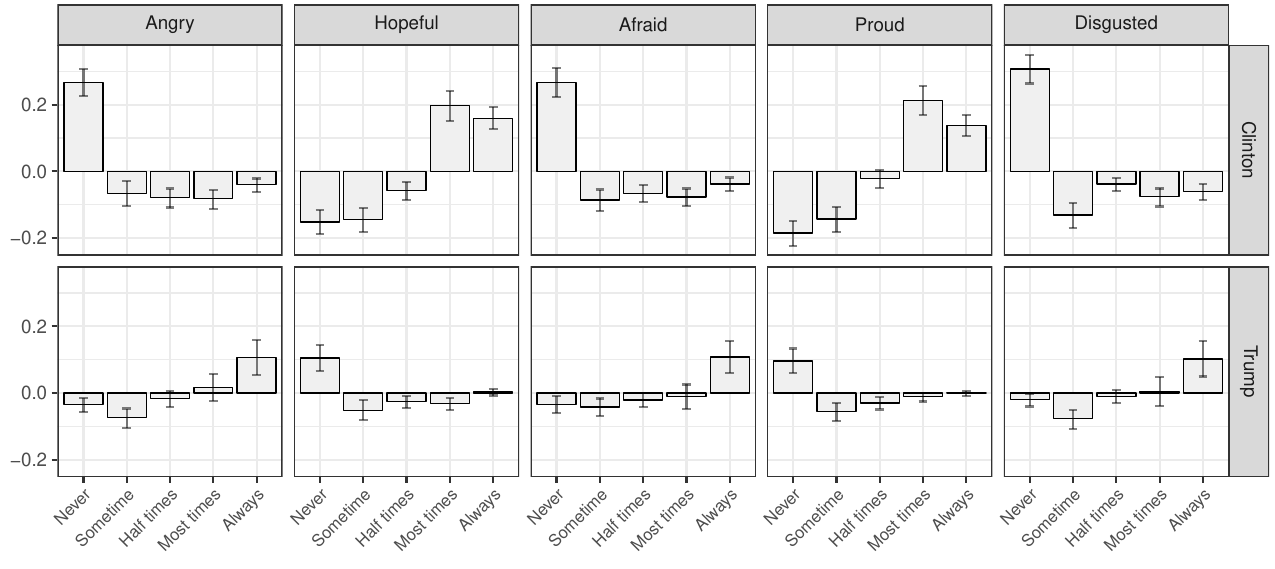} }
\subfigure[Assessments of different personality traits characterizing Hillary Clinton and Donald Trump in the two groups of voters.]{\includegraphics[width=1.01\textwidth]{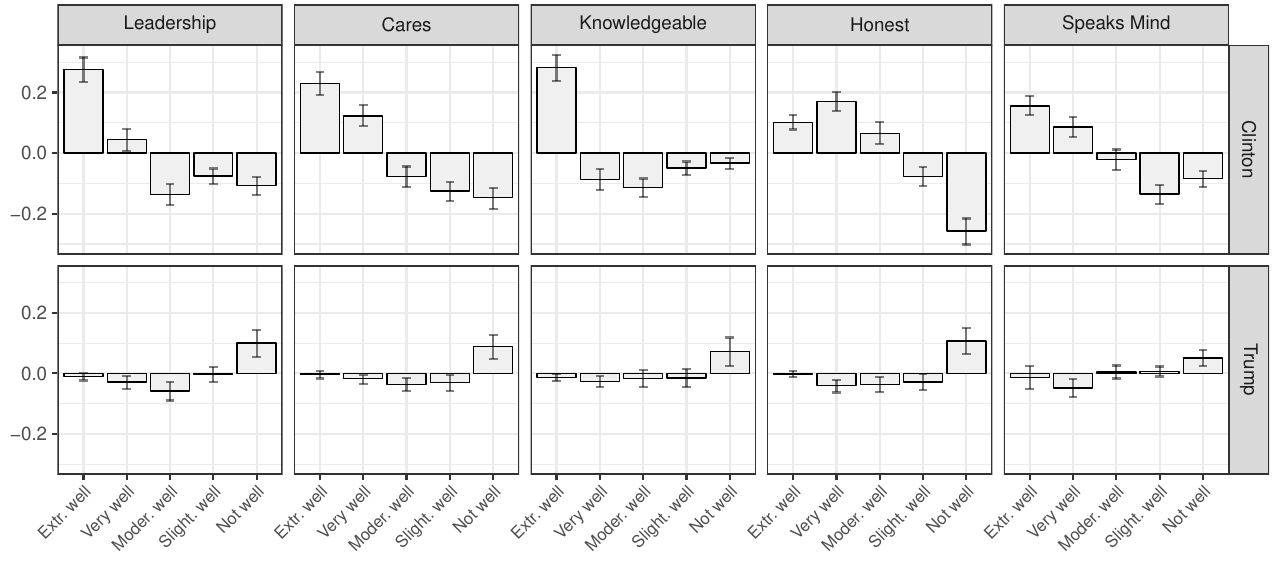} }
\caption{Posterior mean (gray bars), and $0.95$ credible intervals (gray segments) of the difference $\boldsymbol{\pi}_{Y_j\mid X=1}-\boldsymbol{\pi}_{Y_j\mid X=2}$ between the marginal probability mass functions of each qualitative variable in the groups of voters who chose Hillary Clinton and Bernie Sanders, respectively, during the 2016 Democratic Presidential primaries.}
\label{fig:inference_cpp}
\end{figure*}

Results from posterior inference offer interesting insights on group differences in voters opinions with $\hat{\mbox{pr}}\{H_1 \mid (\boldsymbol{y}_1,x_1), \ldots, (\boldsymbol{y}_n,x_n)\}>0.95$ providing  strong evidence of changes in opinions between Hillary Clinton and Bernie Sanders voters. To assess the robustness of this result, we also performed posterior inference based on datasets that randomly matched the observed voting preferences for Hillary Clinton or Bernie Sanders, with a corresponding vector of evaluations on the $p$ items, effectively removing the possibility of a dependence between ${\boldsymbol Y}$ and $X$. In 10 of these trials we always obtained $\hat{\mbox{pr}}\{H_1 \mid (\boldsymbol{y}_1,x_1), \ldots, (\boldsymbol{y}_n,x_n)\}\approx 0$, as expected.

\begin{figure}[h!]
    \centering
    \includegraphics[width=16.6cm]{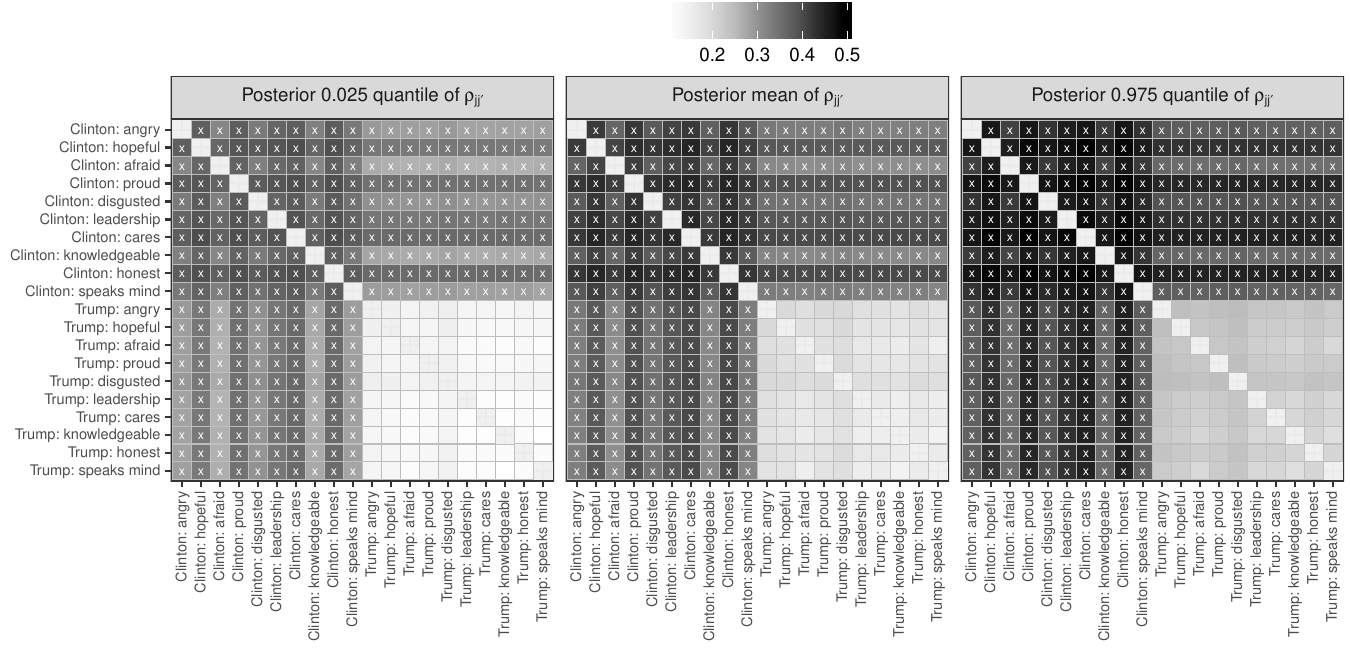}
 \caption{Posterior mean,  $0.025$ posterior quantile, and $0.975$ posterior quantile, for the Cramer's \textsc{v} coefficients $\rho_{jj'}$, measuring group differences in each bivariate. The \texttt{x} symbols denote those pairs of variables whose bivariates are declared to change across groups according to the proposed local tests---i.e. $\hat{\mbox{pr}}\{\rho_{jj'}>0.2\mid(\boldsymbol{y}_1,x_1), \ldots, (\boldsymbol{y}_n,x_n)\}>0.95$. }
    \vspace{-6pt}
  \label{fig:app2}
\end{figure}

The multiple local tests on the marginals  interestingly suggest that the above global variations are attributable to different feelings and opinions on Hillary Clinton. Evaluations of Donald Trump instead do not differ across groups with $\mbox{max}_{j\in \mathcal{J}_{\textsc{D}}}[\hat{\mbox{pr}}\{\rho_j>0.2\mid (\boldsymbol{y}_1,x_1), \ldots, (\boldsymbol{y}_n,x_n)\}]=0.078$, where $\mathcal{J}_{\textsc{D}}$ denotes the set of indices for the variables characterizing feelings and opinions on Donald Trump.  Figure \ref{fig:inference_cpp} clarifies these findings by summarizing the posterior distribution of the difference $\boldsymbol{\pi}_{Y_j\mid X=1}-\boldsymbol{\pi}_{Y_j\mid X=2}$ between the probability mass functions  characterizing the feelings and opinions on Hillary Clinton and Donald Trump, in Hillary Clinton and Bernie Sanders voters, respectively. Leveraging Proposition \ref{prop3}, these quantities are defined as $\pi_{Y_j\mid X=x}(y_j)=\sum_{h=1}^H \nu_{hx}\pi_{hj}(y_j)$, for each $y_{j} \in (1, \ldots, 5)$. Consistent with the local tests, the opinions on Donald Trump remain mostly constant across the two groups, whereas those for Hillary Clinton change. According to Figure \ref{fig:inference_cpp}, these group differences are reasonably due to more negative feelings and opinions expressed by Bernie Sanders voters on Hillary Clinton.

As shown in Figure  \ref{fig:app2}, the changes in the marginals induce also evident group differences in the probability mass function for pairs of variables including at least one assessment on Hillary Clinton. When studying the block of items related to the feelings and opinions on Donald Trump, we do not observe, instead, evidence of group differences in the bivariates. This is an interesting finding, which suggests that the democratic voters share the same joint opinions on the Republican candidates, and express their preference during the primaries mostly based on evaluations of the Democratic candidate, rather than considering their opinions on the potential Republican competitor in the subsequent Presidential elections. Indeed, we applying the model and methodologies described in Section \ref{sec2} only to the vector of items ${\boldsymbol Y}_{\mathcal{J}_{\textsc{D}}}$ measuring feelings and opinions on Donald Trump, we obtained a posterior probability for the global alternative $\hat{\mbox{pr}}\{H_1 \mid (\boldsymbol{y}_1,x_1), \ldots, (\boldsymbol{y}_n,x_n)\}\approx 0$, effectively proving the absence of group differences in the Republican candidate assessments, not only in the marginals and the bivariates, but also in higher-order combinations of items.

\section{Discussion}
Motivated by recent political election studies providing multivariate categorical data on voters opinions and preferences for Presidential candidates, we have developed a novel methodology for testing of group differences in multivariate categorical data at different scales. The proposed procedures rely on a single statistical model based on tensor factorizations, thereby allowing inference and testing on several underlying structures, within a coherent methodological framework. Although this goal can be also accomplished in log-linear models, the proposed group-dependent mixtures of tensor factorizations substantially reduce dimensionality and provide tractable testing procedures, while crucially preserving flexibility---as proved in theoretical studies. These key properties are directly related to the effective borrowing of information within the mixture representation, which additionally induces dependence among the different tests, thus allowing improved power compared to separate univariate tests. Taking a Bayesian approach to inference, we additionally incorporate adaptive selection of the model dimension, and automatic multiplicity control via carefully specified priors. The simulation studies, and the real-data application provide empirical guarantee of the above properties, and highlight improved performance when compared to popular alternatives.

The proposed methods are applicable in broad settings, including unordered and ordered multivariate categorical data. Indeed, in the motivating real-data application there is a natural ordering among the categories of the observed items, which may motivate inclusion of additional structure to incorporate  order restrictions \citep[e.g.][]{agresti2001}. Although  these properties can be easily incorporated within the multinomial kernels in \eqref{eq4}, we avoided additional complications to maintain the model general and fully flexible. In fact, there is no guarantee---a priori---that the ordering in the categories is translated into order restrictions for the probabilistic generative mechanisms of the associated  variables. Another promising direction of research is to incorporate additional dimensionality reduction in equation \eqref{eq4}. In particular, although the proposed statistical model massively reduces the number of parameters compared to log-linear representations, inference may be still cumbersome in large and sparse tensors. A possibility to address this issue is to exploit additional sparsity, by adapting representation \eqref{eq4} to incorporate the recently developed  sparse--\textsc{parafac} model proposed in \citet{Zhou2014}.


\end{document}